\documentclass[11pt]{article}
\usepackage[a4paper, margin=3cm]{geometry}
\usepackage[utf8]{inputenc}
\usepackage[english]{babel}

\usepackage[hidelinks]{hyperref}
\usepackage{url}
\usepackage{graphicx}
\usepackage{amsmath,amsthm,amsfonts}
\usepackage{tikz}
\usepackage{array}
\usepackage{widetable}
\usepackage[labelfont=bf]{caption}
\usepackage{natbib}
\usepackage{authblk}

\title{On Rational Delegations in Liquid Democracy
	\thanks{This paper (without Appendix) appears in the proceedings of AAAI'19.
	We are indebted to the anonymous reviewers of IJCAI/ECAI'18 and  AAAI'19 for many helpful comments on earlier versions of this paper. We are also grateful to the participants of the LAMSADE seminar at Paris Dauphine University, and the THEMA seminar at University Cergy-Pontoise where this work was presented, for many helpful comments and suggestions. 
	Daan Bloembergen has received funding in the framework of the joint programming initiative ERA-Net Smart Energy Systems' focus initiative Smart Grids Plus, with support from the European Union's Horizon 2020 research and innovation programme under grant agreement No 646039. Davide Grossi was partially supported by EPSRC under grant EP/M015815/1. Martin Lackner was supported by the European Research Council (ERC) under grant number 639945 (ACCORD) and by the Austrian Science Foundation FWF, grant P25518 and Y698.\smallskip\footnoterule}}

\author[1]{Daan Bloembergen}
\author[2]{Davide Grossi}
\author[3]{Martin Lackner}
\affil[1]{Centrum Wiskunde \& Informatica, Amsterdam, The Netherlands}
\affil[2]{Bernoulli Institute, University of Groningen, The Netherlands}
\affil[3]{TU Wien, Vienna, Austria}
\affil[ ]{\normalfont{d.bloembergen@cwi.nl, d.grossi@rug.nl, lackner@dbai.tuwien.ac.at}}

\newtheorem{theorem}{Theorem}

\newtheorem{fact}{Fact}

\newtheorem{example}{Example}
\newtheorem{proposition}{Proposition}

\DeclareMathOperator*{\argmax}{arg\,max}

\newcommand{\tuple}[1]{\left\langle #1 \right\rangle} % tuple
\newcommand{\set}[1]{\left\{ #1 \right\}} % set

\newcommand{\aac}{\bar{q}^*}
\newcommand{\sw}{\mathsf{SW}}
\newcommand{\g}{\mathsf{G}}
\newcommand{\NE}{\textit{NE}}
\newcommand{\PoA}{\textit{PoA}}

\renewcommand{\eqref}[1]{Equation~(\ref{#1})}

\newcolumntype{P}[1]{>{\centering\arraybackslash}p{#1}}

\newcommand{\note}[2][Note]{\textcolor{blue}{\textbf{\large [}\colorbox{yellow}{\textbf{#1:}}{\small #2}\textbf{\large ]}}}

\begin{document}

\maketitle

\begin{abstract}
Liquid democracy is a proxy voting method where proxies are delegable. We propose and study a game-theoretic model of liquid democracy to address the following question: when is it rational for a voter to delegate her vote? We study the existence of pure-strategy Nash equilibria in this model, and how group accuracy is affected by them. We complement these theoretical results by means of agent-based simulations to study the effects of delegations on group's accuracy on variously structured social networks.
\end{abstract}

%\clearpage 

\section{Introduction}

Liquid  democracy \citep{blum2016liquid} is an influential proposal in recent debates on democratic reforms in both Europe and the US. Several grassroots campaigns, as well as local parties, experimented with this novel type of decision making procedure. Examples include the German Piratenpartei\footnote{\url{https://www.piratenpartei.de/}} and the EU Horizon 2020 project WeGovNow~\citep{boella2018wegovnow}, which have incorporated the LiquidFeedback\footnote{\url{https://liquidfeedback.org/}} platform in their decision making, as well as grass-roots organizations such as the Democracy Earth Foundation\footnote{\url{https://www.democracy.earth/}}. Liquid democracy is a form of proxy voting \citep{miller1969program,Tullock_1992,Alger_2006,green2015direct,cohensius2017proxy} where, in contrast to classical proxy voting, proxies are delegable (or transitive, or transferable). Suppose we are voting on a binary issue, then each voter can either cast her vote directly, or she can delegate her vote to a proxy, who can again either vote directly or, in turn, delegate to yet another proxy, and so forth. Ultimately, the voters that decided not to delegate cast their ballots, which now carry the weight given by the number of voters who entrusted them as proxy, directly or indirectly.

%\vspace{-0.3cm}

\paragraph{Contribution}
The starting point of our analysis is an often cited feature of liquid democracy:
transitive delegations reduce the level of duplicated effort required by direct voting, by freeing voters from the need to invest effort in order to vote accurately. The focus of the paper is the decision-making problem that voters, who are interested in casting an accurate vote, face between voting directly, and thereby incurring a cost in terms of effort invested to learn about the issue at hand, or delegating to another voter in their network, thereby avoiding costs. 
We define a game-theoretic model, called {\em delegation game}, to represent this type of interaction. We establish pure strategy Nash equilibrium existence results for classes of delegation games, and study the quality of equilibria in terms of the average accuracy they enable for the population of voters, both analytically and through simulations. Proofs of the two main results (Theorems \ref{thm:NE-deterministic} and \ref{thm:delegation-games-NE}) are presented in full, while we provide proofs of the simpler secondary results as an Appendix only.

By means of simulations we also study the effects of different network structures on delegation games in terms of: performance against direct voting, average accuracy and the probability of a correct majority vote, the number and quality of voters acting as ultimate proxies (so-called gurus) and, finally, the presence of delegation cycles.
To the best of our knowledge, this is the first paper providing a comprehensive study of liquid democracy from a game-theoretic angle.

%\vspace{-0.3cm}

\paragraph{Related Work}
Although the idea of delegable proxy was already sketched by \citet{dodgson84principles}, only a few very recent papers have studied aspects of liquid democracy in the (computational) social choice theory \citep{brandt16handbook} literature. \citet{kling2015voting} provide an analysis of election data from the main platform implementing a liquid democracy voting system (Liquid Feedback) for the German Piratenpartei. They focus on network theoretic properties emerging from the structure of delegations---with particular attention to the number of highly influential gurus or `super-voters'. Inspired by their experimental analysis, \citet{golz2018fluid} propose and analyze a variant of the liquid democracy scheme able to restrict reliance on super-voters. \citet{skowron17proportional} study an aspect of the Liquid Feedback platform concerning the order in which proposals are ranked and by which they are brought to the attention of the community. \citet{boldi2011viscous} investigate applications of variants of the liquid democracy voting method (called viscous democracy) to recommender systems.
\citet{brill2018interactive} presents some research directions in the context of liquid democracy.
A general, more philosophical discussion of liquid democracy is provided by \citet{blum2016liquid}.

More directly related to our investigations is the work by \citet{christoff17binary} and, especially, by \citet{kahng18liquid}.
The first paper studies liquid democracy as an aggregator---a function mapping profiles of binary opinions to a collective opinion---in the judgment aggregation and binary voting tradition \citep{Grossi_2014,endriss16judgment}.
The focus of that paper are the unintended effects that transferable proxies may have %in terms of unexpressed votes 
due to delegation cycles, and due to the failure of rationality constraints normally satisfied by direct voting. 
%on the preservation of rationality %constraints when voting on logically %complex issues.

The second paper addresses one of the most cited selling arguments for liquid democracy: delegable proxies guarantee that better informed agents can exercise more weight on group decisions, thereby increasing their quality. Specifically, \citet{kahng18liquid} study the level of accuracy that can be guaranteed by liquid democracy (based on vote delegation with weighted majority) vs. direct voting by majority. Their key result consists in showing that no `local' procedure to select proxies can guarantee that liquid democracy is, at the same time, never less accurate (in large enough graphs)
%(do no harm property) 
and sometimes strictly
%---bounded away from zero---
more accurate 
%(positive gain property) 
than direct voting.
In contrast to their work,
%Unlike \citet{kahng18liquid}, 
we assume that agents incur costs (effort) when voting directly, and on that basis we develop a game-theoretic model. Also, we assume agents aim at tracking their own type rather than an external ground truth, although we do assume such a restriction in our simulations to better highlight how the two models are related and to obtain insights applicable to both.

%%%%%%%%%%%%%%%%%%%%%%%%%%%%%%%%%%%%

%\vspace{-0.1cm}

\section{Preliminaries} %\label{sec:preliminaries}

\subsection{Types, Individual Accuracy and Proximity} %\label{subsec:types}

We are concerned with a finite set of agents (or voters, or players) $N = \set{1, \ldots, n}$ having to decide whether $x = 1$ or $x = 0$. For each agent one of these two outcomes is better, but the agent is not necessarily aware which one. We refer to this hidden optimal outcome as the {\em type of agent $i$} and denote it by $\tau_i \in \set{0,1}$. Agents want to communicate their type truthfully to the voting mechanism, but they know it only imperfectly. This is captured by the \emph{accuracy} $q_i$ of an agent $i$:
$q_i$ determines the likelihood that, if $i$ votes directly, she votes according to her type $\tau_i$.
We assume that an agent's accuracy is always $\geq 0.5$, i.e., at least as good as a coin toss.

We distinguish two settings depending on whether the agents' types are deterministic or probabilistic. 
A {\em deterministic type profile} $T=\tuple{\tau_1, \ldots, \tau_n}$ simply collects each agent's type.
In {\em probabilistic type profiles} types are independent random variables drawn according to a distribution ${\mathbb P}$. Given a probabilistic type profile, the likelihood that any two agents $i,j\in N$ are of the same type is called the {\em proximity} $p_{i,j}$ where
%as the likelihood that they have the same type, i.e., 
$p_{i,j}={\mathbb P}(\tau(i)=\tau(j)) = {\mathbb P}(\tau(i) = 1)\cdot {\mathbb P}(\tau(j) = 1) +  (1- {\mathbb P}(\tau(i) = 1)) \cdot  (1- {\mathbb P}(\tau(j) = 1))$. In the probabilistic setting we assume agents know such value although, importantly, they do not know ${\mathbb P}$.
In a deterministic type profile, we have $p_{i,j}= 1$ if $\tau_i=\tau_j$ and $p_{i,j}= 0$ otherwise. Following standard equilibrium theory, our theoretical results assume agents act as if they have access to the accuracy of each agent. More realistically, in our simulations we assume agents have access to such information only with respect to neighbors on an underlying interaction structure.

\subsection{Interaction Structure and Delegations} %\label{subsec:structure}

Agents are nodes in a network (directed graph) represented by a relation $R \subseteq N^2$. For $i \in N$ the neighborhood of $i$ in $\langle N, R \rangle$ is denoted $R(i)$, i.e., the agents that are directly connected to $i$. 
%Further let $R^*$ denote the reflexive and transitive closure of $R$, i.e., $(i,j)\in %R^*$ if and only if there is a path in $R$ from $i$ to $j$. In our simulations we will %assume $N^2 \subseteq R^*$, i.e., the network consists of a single connected component.
Agents have the choice of either voting themselves, thereby relying solely on their own accuracy, or delegating to an agent in their neighborhood. A {\em delegation profile} is a vector ${\bf d} = \tuple{d_1, \ldots, d_n} \in N^n$. Given a delegation profile ${\bf d}$ we denote by $d_i$ the \emph{proxy} selected by $i$ in ${\bf d}$. Clearly a delegation profile can be viewed as a functional graph on $N$ or, equivalently, as a map in ${\bf d}: N \to N$ where ${\bf d}(i) = d_i$. When the iterated application of ${\bf d}$ from $i$ reaches a fixed point we denote such fixed point as $d^*_i$ and call it $i$'s {\em guru} (in ${\bf d}$).
In the following, we write $N^*$ to denote the set of voters whose delegation does not lay on a path ending on a cycle, i.e., the set of voters $i$ for which $d^*_i$ exists. 
%We will regularly refer to the concatenation of an effort profile $\e$ with a delegation profile ${\bf d}$ simply as a {\em profile} ${\bf d}$.
We write ${\bf d}' = ({\bf d}_{-i},j)$ as a short form for ${\bf d}' = \tuple{d_1, \ldots, d_{i-1}, j, d_{i+1},\dots, d_n} $.

As agents may only be able to observe and interact with their direct network neighbors, structural properties of the interaction network may play a role in the model dynamics. In our simulations we will focus on undirected graphs (that is, $R$ will be assumed to be symmetric, as social ties are normally mutual) consisting of one single connected component (that is, $N^2$ is included in the reflexive transitive closure of $R$).  
Under these assumptions, we consider four typical network structures that are well represented in the literature on social network analysis (cf. \citealt{jackson08social}): 1) the {\em random} network, in which each pair of nodes has a given probability of being connected \citep{erdos1959random}; 2) the {\em regular} network, in which all nodes have the same degree; 3) the {\em small world} network, which features a small average path length and high clustering \citep{watts1998collective}; and 4) the {\em scale free} network, which exhibits a power law degree distribution \citep{barabasi1999emergence}.\footnote{Although random and regular graphs are not generally applicable to real-world settings, they serve as a useful baseline to illustrate the effects of network structure on delegations.
%for comparison and bridge with the work of \citet{kahng18liquid}, %who consider fully connected graphs in their work.
}

\section{A Model of Rational Delegations} %\label{sec:model}

\subsection{Individual Accuracy under Delegable Proxy}

Each agent $i$ has to choose between two options: either to vote herself with accuracy $q_i$ or to delegate, thereby inheriting the accuracy of another voter (unless $i$ is involved in a delegation cycle).
These choices are recorded in the delegation profile ${\bf d}$ and can be used to compute the individual accuracy for each agent $i\in N^*$ as follows:
\begin{align} \label{eq:accuracy}
    \small
  \hspace{-0.3cm}  q^*_i({\bf d}) = 
    \begin{cases} 
    q_{d^*_i}\cdot p_{i,d^*_i} 
    %& {} \\ \quad 
    + (1-q_{d^*_i}) \cdot (1-p_{i,d^*_i}) & \text{if }i\in N^* \\
    0.5 &\text{if }i\notin N^*
    \end{cases} 
\end{align}
%\begin{align} \label{eq:accuracy}
%    q^*_i({\bf d}) = 
%    q_{d^*(i)}\cdot p_{i,d^*(i)} + (1-q_{d^*(i)})\cdot (1-p_{i,d^*(i)}),
%\end{align}
In \eqref{eq:accuracy}
$i$'s accuracy equals the likelihood that $i$'s guru has the same type and votes accurately plus the likelihood that $i$'s guru has the opposite type and fails to vote accurately. Note that if $i$ votes directly, i.e., $d_i=i$, then $q^*_i({\bf d}) = q_i$. Observe that if $i$'s delegation leads to a cycle ($i\notin N^*$), $i$'s accuracy is set to $0.5$. The rationale for this assumption is the following. If an agent delegates into a cycle, even though she knows her own accuracy and she actively engages with the voting mechanism by expressing a delegation, she fails to pass information about her type to the mechanism. No information is therefore available to decide about her type. 

%, as this voter does not participate in %the decision process and hence does not %influence the decision.
%Without prior knowledge about other voters' preferences, $i$'s accuracy equals that of a fair coin toss.

It may be worth observing that, by \eqref{eq:accuracy}, in a deterministic type profile we have that $p_{i,j}\in\{0,1\}$ and therefore $i$'s accuracy reduces to: $q_{d^*_i}$ if $i\in N^*$ and $\tau(i)=\tau(d^*_i)$; $1-q_{d^*(i)}$ if $i\in N^*$ and $\tau(i)\neq\tau(d^*_i)$; and $0.5$ if $i\notin N^*$.

\iffalse
\begin{align} \label{eq:accuracy_deterministic}
    q^*_i({\bf d}) = \begin{cases}
    q_{d^*(i)} &\text{if }i\in N^* \text{ and }\tau(i)=\tau(d^*_i)\text{,}\\
    1-q_{d^*(i)} & \text{if }i\in N^* \text{ and }\tau(i)\neq\tau(d^*_i)\text{,}\\
    0.5 &\text{if }i\notin N^*.
    \end{cases}
\end{align}
Let us immediately establish a basic fact about delegations in our model: assuming a weak rationality condition for all agents, an agent can identify some delegations that have a guaranteed positive impact on her accuracy, even if her proxy further delegates to agents she has no information about.
\fi

Before introducing our game theoretic analysis, we make the following observation. Agents have at their disposal an intuitive strategy to improve their accuracy: simply delegate to a more accurate neighbor. We say that a delegation profile ${\bf d}$ is \emph{positive} if for all $j\in N$ either $d_j=j$ or $q^*_j({\bf d})> q_j$.
Furthermore, we say that a delegation from $i$ to a neighbor $j$ is 
\emph{locally positive} if $q_j\cdot p_{i,j} + (1-q_j)\cdot (1-p_{i,j})>q_i$.

\begin{proposition}
Let ${\bf d}$ be a positive delegation profile.
Further, let $s,t\in N$, $d_s=s$, and ${\bf d}'=({\bf d}_{-s},t)$, i.e., agent $s$ votes directly in ${\bf d}$ and delegates to $t$ in ${\bf d}'$.
%has not delegated yet but considers to delegate to $t$.
If the delegation from $s$ to $t$ is locally positive, then 
%$q^*_s({\bf d}') > q^*_s({\bf d}) = q_s$ and hence 
${\bf d}'$ is positive (proof in Appendix~\ref{appendix:proof}).
\label{prop:positive}
\end{proposition}
%\note[Daan]{Should we refer to the %supplementary material for the proof?}
However, locally positive delegations do not necessarily correspond to optimal delegations.
%Optimal delegations, however, cannot be %identified with local knowledge. 
This can be easily seen
in an example where agent $i$ is not a neighbor of a very competent agent $j$, but would have to delegate via an intermediate agent $k$ (who delegates to $j$). If this intermediate agent $k$ has a lower accuracy than $i$, then the delegation from $i$ to $k$ would not be locally positive, even though it is an optimal choice. So utility-maximization may require information which is inherently non-local (accuracy of `far' agents). %We will come back to this feature of liquid democracy later in Sections~\ref{subsec:equilibria} and~\ref{sec:simulation}.  

\subsection{Delegation Games} %\label{subsec:delegationgames}

We assume that each agent $i$ has to invest an effort $e_i$ to manifest her accuracy $q_i$.
If she delegates, she does not have to spend effort.
%In other words, The effort $e_i$ determines how much $i$ has to invest so that he achieves his accuracy $q_i$.
Agents aim therefore at maximizing the trade-off between the accuracy they can achieve (either by voting directly or through proxy) and the effort they spend. Under this assumption, the binary decision set-up with delegable proxy we outlined above can be used to define a natural game---called {\em delegation game}---$G = \tuple{N, {\mathbb P}, R, \Sigma_i, u_i}$, with $i \in N$, where $N$ is the set of agents, ${\mathbb P}$ is the (possibly degenerate) distribution from which the types of the agents in $N$ are drawn, $R$ the underlying network as defined above, $\Sigma_i\in  N$ is the set of strategies of agent $i$ (voting, or choosing a specific proxy), and 
\begin{align}
    u_i({\bf d}) = 
    \begin{cases} q^*_i({\bf d})  &\text{if } d_i\neq i\\
    q_i - e_i &\text{if } d_i= i\\
   % 0.5 &\text{if }i\notin N^*
    \end{cases} \label{eq:utility}
\end{align}
is agent $i$'s utility function. The utility $i$ extracts from a delegation profile equals the accuracy she inherits through proxy or, if she votes, her accuracy minus the effort spent.\footnote{No utility is accrued for gaining voting power in our model.} In delegation games we assume that $q_i-e_i\geq 0.5$ for all $i\in N$. This is
because if $q_i-e_i < 0.5$, then $i$ would prefer a random effortless choice over taking a decision with effort.

A few comments about the setup of \eqref{eq:utility} are in order. First of all, as stated earlier, we assume agents to be truthful. They do not aim at maximizing the chance their type wins the vote, but rather to relay their type to the mechanism as accurately as possible.\footnote{Notice however that our modeling of agents' utility remains applicable in this form even if agents are not truthful but the underlying voting rule makes truthfulness a dominant strategy---such as majority in the binary voting setting used here.} 
%This assumption has, in our view, two benefits: it makes the model parsimonious; and it allows us to show that interesting strategic behavior can arise, once voters have an option to delegate, even with strategically naive (in the voting games sense) agents.
Secondly, notice that the utility an agent extracts from a delegation profile may equal the accuracy of a random coin toss when the agent's delegation ends up into a delegation cycle (cf. \eqref{eq:accuracy}). If this happens the agent fails to relay information about her type, even though she acted in order to do so. This justifies the fact that $0.5$ is also the lowest payoff attainable in a delegation game. 

%\smallskip

%This happens when the agent's delegation %ends up into a delegation cycle. This  %Such value is also the lowest payoff %attainable, capturing the fact the the %agent fails to relay information about her %type even though she acted in order to do %so.

The following classes of delegation games will be used in the paper: 
%{\em worthwile voting} games, where for %each $i \in N$ $q_i - e_i > 0.5$, that %is, voters prefer to vote directly rather %than randomly; 
games with {\em deterministic profiles}, i.e., where ${\mathbb P}$ is degenerate and all players are assigned a crisp type from $\set{0,1}$; {\em homogeneous} games, where all players have the same (deterministic) type;\footnote{This is the type of interaction studied, albeit not game-theoretically, by \citet{kahng18liquid} and normally assumed by jury theorems \citep{grofman83thirteen}.} and {\em effortless voting} games, where for each $i \in N$ we have $e_i = 0$.

As an example, a homogeneous game in matrix form is given in Table~\ref{table:dgame}, where $N = \set{1,2}$, $R = N^2$ and the distribution yields the deterministic type profile $T = \tuple{1,1}$. Interestingly, if we assume that $q_i - e_i > 0.5$ with $i \in \set{1,2}$, and that\footnote{We use here the usual notation $-i$ to denote
 $i$'s opponent.} $q_{-i} > q_i - e_i$ (i.e., the opponent's accuracy is higher than the player's individual accuracy minus her effort), then the game shares the ordinal preference structure of the class of anti-coordination games: players need to avoid coordination on the same strategy (one should vote and the other delegate), with two coordination outcomes (both players voting or both delegating) of which the second (the delegation cycle) is worst for both players. Notice that, were the underlying network not complete (i.e., $R \subset N^2$), the matrix would be shrunk by removing the rows and columns corresponding to the delegation options no longer available.

% \begin{table}[tb]
% 	\begin{center}
% 		\begin{game}{2}{2}
% 			& vote                  & delegate (to $1$) \\
% 			vote & $q_1-e_1, q_2-e_2$              & $q_1-e_1,q_1 $\\
% 			delegate (to $2$)      & $q_2,q_2-e_2 $      & $0.5,0.5$
% 		\end{game}
% 	\end{center}
% 	\caption{A two-player delegation game. The row player is agent $1$ and the column player is agent $2$.}
% 	\label{table:dgame}
% 	%\vspace{-0.3cm}
% \end{table}

\bgroup
\def\arraystretch{1.1}
\begin{table}[tb]
	\centering
	\begin{tabular}{rP{2.2cm}|P{2.2cm}|}
		 & \multicolumn{1}{c}{vote} & \multicolumn{1}{c}{delegate (to $1$)} \\ \cline{2-3}
		\multicolumn{1}{r|}{vote} & $q_1-e_1, q_2-e_2$ & $q_1-e_1,q_1 $ \\ \cline{2-3}
		\multicolumn{1}{r|}{delegate (to $2$)} & $q_2,q_2-e_2 $ & $0.5,0.5$ \\ \cline{2-3}
	\end{tabular}
	\caption{A two-player delegation game. The row player is agent $1$ and the column player is agent $2$.}
	\label{table:dgame}
	%\vspace{-0.3cm}
\end{table}
\egroup

%We will see in the following sections %that 
The introduction of effort has significant consequences on the delegation behavior of voters, and we will study it in depth in the coming sections. It is worth noting immediately that the assumptions of Proposition~\ref{prop:positive} no longer apply, since agents may prefer to make delegations that are not locally positive due to the decreased utility of voting directly.

\subsection{Existence of Equilibria in Delegation Games} %\label{subsec:equilibria}

In this section we study the existence of pure strategy Nash Equilibria (NE) in two classes of delegation games. NE describe how ideally rational voters would resolve the effort/accuracy trade-off. Of course,  such voters have common knowledge of the delegation game structure---including, therefore, common knowledge of the accuracies of `distant' agents in the underlying network. Our simulations will later lift some of such epistemic assumptions built into NE. 

\paragraph{Deterministic Types}

In the following we provide a NE existence result for games with deterministic type profiles.
\begin{theorem}
Delegation games with deterministic type profiles always have a (pure strategy) NE.\label{thm:NE-deterministic}
\end{theorem}
\begin{proof}
First of all, observe that since the profile is deterministic, for each pair of agents $i$ and $j$, $p_{i,j} \in \set{0,1}$. 
%Hence, accuracies of voters follow from %\eqref{eq:accuracy}. 
The proof is by construction. 
First, we partition the set of agents $N$ into $N_1 = \set{i \in N \mid \tau(i) = 1}$ and $N_0 = \set{i \in N \mid \tau(i) = 0}$.
We consider these two sets separately; without loss of generality let us consider $N_1$.
Further we consider the network $R_1=\{(i,j)\in N_1\times N_1: (i,j)\in R\}$.
Since $(N_1,R_1)$ can be seen as a directed graph, we can partition it into Strongly Connected Components (SCCs).
If we shrink each SCC into a single vertex, we obtain the condensation of this graph; note that such a graph is a directed acyclic graph (DAG).
We construct a delegation profile ${\bf d}$ by traversing this DAG bottom up, i.e., starting with leaf SCCs.

Let $S\subseteq N_1$ be a set of agents corresponding to a leaf SCC in the condensation DAG.
We choose an agent $i$ in $S$ that has maximum $q_i-e_i$.
Everyone in $S$ (including $i$) delegates to $i$.
Now let $S\subseteq N_1$ be a set of agents corresponding to an inner node SCC in the condensation DAG and assume that we have already defined the delegation for all SCCs that can be reached from $S$.
As before, we identify an agent $i\in S$ with maximum $q_i-e_i$.
Further, let $T\subseteq N_1\setminus S$ be the set of all voters $j$ that can be reached from $S$ in $(N_1,R_1)$, and for which $q^*_j>q_i-e_i$.
We distinguish two cases. 
(i) If $T\neq \emptyset$, then we choose an agent $k\in T$ with $q^*_k=\max_{j\in T} q^*_j$ and all agents in $S$ directly or indirectly delegate to $k$.
(ii) If $T=\emptyset$, all agents in $S$ delegate to $i$.
This concludes our construction (as for $N_0$ the analogous construction applies); let ${\bf d}$ be the corresponding delegation profile.

It remains to verify that this is indeed a NE:
Let $i$ be some agent in an SCC $S$, and, without loss of generality, let $i\in N_1$.
Observe that since we have a deterministic profile, if agent $i$ changes her delegation to $j$, then $i$'s utility changes to $q_j^*({\bf d})$ if $i\in N_1$ and $1-q_j^*({\bf d})$ if $i\in N_0$.
First, note that for all agents $k\in N$, $q^*_k({\bf d})\geq q_k-e_k \geq 0.5$.
Hence, we can immediately exclude that for agent $i$ delegating to an agent in $j\in N_0$ is (strictly) beneficial, as it would yield an accuracy of at most $1-q^*_j\leq 0.5$.
Towards a contradiction assume there is a beneficial deviation to an agent $j\in N_1$, i.e., there is an agent $j\in R(i)\cap N_1$ with $q^*_j({\bf d}) > q_i^*({\bf d})$.
Let us now consider the three cases: (1) $d_i=i$, (2) $d^*_i\in S$ but $d_i\neq i$, and (3) $d^*_i\notin S$.
In case (1), everyone in $S$ delegates to $i$. Hence, if $j\in S$, a cycle would occur yielding a utility of $0.5$, which is not sufficient for a beneficial deviation.
If a delegation to $j\notin S$ is possible but was not chosen, then by construction $q^*_j \leq q_i-e_i$ and hence this deviation is not beneficial. We conclude that in case (1) such an agent $j$ cannot exist.
In case (2), everyone in $S$ delegates to $d^*_i$. Hence, if $j\in S$, then $d^*_j=d^*_i$, a contradiction. If $j\notin S$, the same reasoning as before applies and hence also here we obtain a contradiction.
In case (3), by construction, $d^*_i\notin S$ had been chosen to maximise accuracy, hence $j\in S$.
Since for all $k\in S$, $d_k^*=d_i^*$, only a deviation to $i$ itself can be beneficial, i.e., $j=i$. However, since $d^*_i$ was chosen because $q^*_{d^*(i)} > q_i-e_i$, no beneficial deviation is possible.
%We have indeed constructed a NE.
\end{proof}
%A direct corollary of this result is 
\noindent
It follows that also homogeneous games always have NE. 
%Furthermore, by the construction in the %proof we obtain as corollary that NE in %delegation games with deterministic %profiles can be found in polynomial time.

%\begin{corollary}
%A (pure strategy) NE of delegation games %with deterministic profiles can be found %in polynomial time.
%\end{corollary}

\paragraph{Effortless Voting} %\label{sec:effortless}

%In the following we assume that $e_i=0$ for all $i\in N$; we refer to this as {\em effortless voting}.
Effortless voting ($e_i=0$ for all $i\in N$) is applicable whenever effort is spent in advance of the decision and further accuracy improvements are not possible.
%\footnote{E.g., in systems of sensors with fixed error models.}

%can no longer be changed during the %delegation game. 
%This applies for instance to political %elections, where expertise has to be %accrued over a longer period of time. 
%Since the effort each agents spends is fixed, it can be assumed without loss of generality that this effort is $0$.

\begin{theorem}
Delegation games with effortless voting always have a (pure strategy) NE.\label{thm:delegation-games-NE}
\end{theorem}
\begin{proof}
We prove this statement by showing that the following procedure obtains a NE:
We start with a strategy profile in which all players vote directly, i.e., player $i$'s strategy is $i$. 
Then, we iteratively allow players to choose their individual best response strategy to the current strategy profile.
Players act sequentially in arbitrary order.
If there are no more players that can improve their utility by changing their strategy, we have found a NE.
%In the following we will show that such a convergence indeed occurs.
We prove convergence of this procedure by showing that a best response that increases the player's utility never decreases the utility of other players.

%First, we show by induction that a best response that increases the player's utility never decreases the utility %of other players.
We proceed by induction.
Assume that all previous best responses have not reduced any players' utility (IH).
Assume player $i$ now chooses a best response that increases her utility.
Let ${\bf d}$ be the delegation profile; further, let $d^*_i=s$.
By assumption, $i$'s utility started with $q_i-e_i=q_i$ and has not decreased since, i.e., $u_i({\bf d})\geq q_i$.
Since $i$'s best response strictly increases $i$'s utility, it cannot be a delegation to herself. 
So let a delegation to $j \neq i$ be $i$'s best response and consider profile ${\bf d}'=({\bf d}_{-i},j)$.
%We define ${\bf d}'=({\bf d}_{-i},j)$, the delegation profile.
Further, let $d_j^*=t$, i.e., $i$ now delegates to $j$ and by transitivity to $t$, i.e., $d_i'^*=d_j'^*=t$.
%\note[Davide]{Should $j$ be $t$? Unless the interpretation is that $i$ delegates to $j$ directly, and transitively to %$t$}
%\note[Martin]{Indeed, what I meant is that $i$ delegates to $j$ directly, and transitively to $t$. I think the %further notation makes sense with this definition.}
Let $k$ be some player other than $i$.
We define the delegation path of $k$ as the sequence $({\bf d}(k), {\bf d}({\bf d}(k)), {\bf d}({\bf d}({\bf d}(k))),\dots)$.
If $k$'s delegation path does not contain $i$, then $k$'s utility remains unchanged, i.e., $u_k({\bf d}')\geq u_k({\bf d})$.
If $k$'s delegation path contains $i$, then $k$ now delegates by transitivity to $t$, i.e., we have $d_k^*=s$ and $d_k'^*=t$.
%\note[Davide]{Here I think $d_k'^*=t$ should probably be $d_k'^*=d_t'^*$, that is, since $i$ delegates to $t$ now $i$'s guru, and $k$'s, is $t$'s guru.}
By \eqref{eq:utility}, we have 
\begin{align}
&u_k({\bf d}) = q_s\cdot p_{k,s} + (1-q_{s})\cdot (1-p_{k,s}) \quad\text{ and}\label{eq:uked}\\
&u_k({\bf d}') = q_t\cdot p_{k,t} + (1-q_{t})\cdot (1-p_{k,t}).\label{eq:uke'd'}
\end{align}
We have to show that $k$'s utility does not decrease, i.e., $u_k({\bf d}') \geq u_k({\bf d})$, under the assumption that $i$ chooses a best response, i.e., $u_i({\bf d}') > u_i({\bf d})$,
% \note[Daan]{shouldn't this be a strict increase $u_i({\bf d}') > u_i({\bf d})$?} 
 with:
\begin{align}
&u_i({\bf d}) = q_s\cdot p_{i,s} + (1-q_{s})\cdot (1-p_{i,s}) \quad\text{ and}\label{eq:uied}\\
&u_i({\bf d}') = q_t\cdot p_{i,t} + (1-q_{t})\cdot (1-p_{i,t}).\label{eq:uie'd'}
\end{align}
In the following we will often use the fact that, for $a,b\in[0,1]$, if $ab+(1-a)(1-b)\geq 0.5$, then either $a,b\in[0,0.5]$ or $a,b\in[0.5,1]$.
By IH, since accuracies are always at least 0.5, it holds that $u_i({\bf d})\geq q_i\geq 0.5$ and by \eqref{eq:uied} we have $q_s \cdot p_{i,s} + (1-q_s)\cdot (1-p_{i,s}) \geq 0.5 $ and hence $p_{i,s}\geq 0.5$.
Analogously, \eqref{eq:uked} implies that $p_{k,s}\geq 0.5$.

Furthermore, we use the fact that 
\begin{equation}
%\begin{split}
%p_{k,i} & = p_{k,s}p_{s,i}+(1-p_{k,s})(1-p_{s,i}) \\ & \qquad {} + (- 2 (2x_k-1) \cdot (2 x_i - 1) \cdot (x_s - 1) x_s) 
%\end{split}
p_{k,i} = p_{k,s}p_{s,i}+(1-p_{k,s})(1-p_{s,i}) + (- 2 (2x_k-1) \cdot (2 x_i - 1) \cdot (x_s - 1) x_s)
\label{eq:fact1}
\end{equation}
where $x_j = {\mathbb P}(\tau(j) = 1)$ for $j \in \set{k,i,s}$. Observe that, by the definition of utility in \eqref{eq:utility}, the assumptions made on ${\bf d}$ and ${\bf d}'$, and the fact that for $a,b\in[0,1]$, if $ab+(1-a)(1-b)\geq 0.5$, then either $a,b\in[0,0.5]$ or $a,b\in[0.5,1]$. So we have that either  $x_j \geq 0.5$ for $j \in \set{k, i, s}$, or $x_j \leq 0.5$ for $j \in \set{k, i, s}$. We work on the first case. The other case is symmetric. Let also $\gamma_{k,s,i} = - 2 (2x_k-1) \cdot (2 x_i - 1) \cdot (x_s - 1) x_s$. From the above it follows that $0.5 \geq \gamma_{k,i,s} \geq 0$. Furthermore, given that $p_{i,s} = p_{s,i} \geq 0.5$, we can also conclude that $p_{k,i}\geq 0.5$.
% and therefore that $p_{k,i}\geq 0.5$.
Now by substituting 
\begin{equation*}
%\begin{split}
%p_{k,s} & = p_{k,i}p_{i,s}+(1-p_{k,i})(1-p_{i,s}) \\ & \qquad {} + (\underbrace{- 2 (2x_k -1) \cdot (2 x_s - 1) \cdot (x_i - 1) x_i}_{\gamma_{k,i,s}})
%\end{split}
p_{k,s} = p_{k,i}p_{i,s}+(1-p_{k,i})(1-p_{i,s}) + (\underbrace{- 2 (2x_k -1) \cdot (2 x_s - 1) \cdot (x_i - 1) x_i}_{\gamma_{k,i,s}})
\end{equation*}
in \eqref{eq:uked}, we obtain
\begin{equation}
%\begin{split}
%u_k({\bf d}) & = (2 p_{k,i}-1) (\overbrace{2q_sp_{i,s}-q_s-p_{i,s}+1}^{u_i({\bf d})}) 
%\\ & \qquad \qquad {} + 1 - p_{k,i} + \gamma_{k,i,s}(2q_s - 1). 
%\end{split}
u_k({\bf d}) = (2 p_{k,i}-1) (\overbrace{2q_sp_{i,s}-q_s-p_{i,s}+1}^{u_i({\bf d})}) + 1 - p_{k,i} + \gamma_{k,i,s}(2q_s - 1). 
\label{eq:uked=uied}
\end{equation}
%and therefore $u_k({\bf d}) \geq u_i({\bf d})$ given that .

Similarly, using the appropriate instantiation of \eqref{eq:fact1} for $x_j$ with $j \in \set{k,i,t}$, by substituting 
%\begin{align*}
$
p_{k,i}\cdot p_{i,t}+(1-p_{k,i})(1-p_{i,t}) + \gamma_{k,i,t}
$
for $p_{k,t}$ in \eqref{eq:uke'd'} we obtain
\begin{equation}
%\begin{split}
%u_k({\bf d}') & = (2 p_{k,i}-1)\cdot (\overbrace{2q_tp_{i,t}-q_t-p_{i,t}+1}^{u_i({\bf d}')}) \\ & \qquad \qquad {} +  1 - p_{k,i} + \gamma_{k,i,t}(2q_t - 1).
%\end{split}
u_k({\bf d}') = (2 p_{k,i}-1)\cdot (\overbrace{2q_tp_{i,t}-q_t-p_{i,t}+1}^{u_i({\bf d}')}) +  1 - p_{k,i} + \gamma_{k,i,t}(2q_t - 1).
\label{eq:uke'd'=uie'd'}
\end{equation}
Now observe that, since $p_{k,i}\geq 0.5$ we have that $(2 p_{k,i}-1)\geq 0$. It remains to compare $\gamma_{k,i,s}(2q_s - 1)$ with $\gamma_{k,i,t}(2q_t - 1)$, showing the latter is greater than the former. Observe that both expressions have a positive sign. We use the fact that $ab+(1-a)(1-b) < cd+(1-c)(1-d)$ implies $cd > ab$ under the assumption that $a,b,c,d\in[0.5,1]$. On the basis of this, and given that $u_i({\bf d}')> u_i({\bf d})$, we obtain that $q_s\cdot p_{i,s} < q_t\cdot p_{i,t}$ and therefore that $q_s\cdot x_s < q_t\cdot x_t$, from which we can conclude that
\begin{equation*}
\begin{split}
&\big(\overbrace{- 2 (2x_k -1) \cdot (2 x_s - 1) \cdot (x_i - 1) x_i}^{\gamma_{k,i,s}}\big)\cdot(2q_s - 1) \\
&< \big(\underbrace{- 2 (2x_k -1) \cdot (2 x_t - 1) \cdot (x_i - 1) x_i}_{\gamma_{k,i,t}}\big) \cdot(2q_t - 1).
\end{split}
%(\overbrace{- 2 (2x_k -1) \cdot (2 x_s - 1) \cdot (x_i - 1) x_i}^{\gamma_{k,i,s}})\cdot(2q_s - 1) < (\underbrace{- 2 (2x_k -1) \cdot (2 x_t - 1) \cdot (x_i - 1) x_i}_{\gamma_{k,i,t}}) \cdot(2q_t - 1).
\label{eq:extra}
\end{equation*}
It follows that the assumption $u_i({\bf d}')> u_i({\bf d})$ (player $i$ chose a best response that increased her utility) together with Equations~(\ref{eq:uked=uied}) and~(\ref{eq:uke'd'=uie'd'}) implies that $u_k({\bf d}')> u_k({\bf d})$ (and {\em a fortiori} that $u_k({\bf d}') \geq u_k({\bf d})$). We have therefore shown that if some player chooses a best response, the utility of other players does not decrease. This completes the proof.
%Finally, observe that there is only a finite (fixed) number of utilities for players.
%As a best response increases the utility of the acting player and all other utilities do not decrease, we infer %that our procedure converges and a NE exists.
\end{proof}
\iffalse
\begin{proposition}
There exist delegation games without (pure strategy) NE.
\end{proposition}
\begin{proof}
Embedding a $3$-player matching penny (for $N \geq 3$) showing that best-response dynamics leads to cycles from every possible profile. 

\note[Davide]{the construction is possible in cases in which proximity is not symmetric, possibly a case we want to rule out.}

\note[Martin]{Is a proof like this still possible in our new setting?}
\end{proof}
\fi
%Like in the case of deterministic %profiles, our proof yields as corollary %that, in delegation games with effortless %voting, NE can be found in polynomial %time.

%of the existence of a NE is constructive %5and can be executed in polynomial time. %Hence:

%\begin{corollary}
%A Nash equilibrium of delegation games %with effortless voting can be found in %polynomial time.
%\end{corollary}
%\begin{proof}
%The procedure used in the proof of Theorem~\ref{thm:delegation-games-NE} requires only polynomial time.
%\end{proof}

\paragraph{Discussion}

The existence of NE in general delegation games remains an interesting open problem. It should be noted that the proof strategies of both Theorems~\ref{thm:NE-deterministic} and~\ref{thm:delegation-games-NE} do not work in the general case. Without a clear dichotomy of type it is not possible to assign delegations for all agents in an SCC (as we do in the proof of Theorem~\ref{thm:NE-deterministic}). And the key property upon which the proof of Theorem \ref{thm:delegation-games-NE} hinges (that a best response of an agent does not decrease the utility of other agents) fails in the general case due to the presence of non-zero effort.
%\footnote{In fact it should be observed that our proof strategy for Theorem \ref{thm:delegation-games-NE} via %best response dynamics fails even in the restricted case of homogeneous games with effort, although our %experimental results indicate that the dynamics does converge in this case too. \label{footnote:homogeneous}}
Finally, it should also be observed that Theorem \ref{thm:delegation-games-NE} (as well as Proposition~\ref{prop:positive})
essentially depend on the assumption that types are \emph{independent} random variables. If this is not the case (e.g., because voters' preferences are correlated), delegation chains can become undesirable.

\begin{example}
Consider the following example with agents $1$, $2$ and $3$. The probability distribution ${\mathbb P}$ is defined as 
%\begin{align*}
${\mathbb P}(\tau(1)=1\wedge \tau(2)=1\wedge \tau(3)=0)=0.45$,
${\mathbb P}(\tau(1)=0\wedge \tau(2)=1\wedge \tau(3)=1)=0.45$, and
${\mathbb P}(\tau(1)=1\wedge \tau(2)=1\wedge \tau(3)=1)=0.1.
$
%\end{align*}
Consequently, $p_{1,2}=0.55$, $p_{2,3}=0.55$, and $p_{1,3}=0.1$.
Let us assume that the agents' accuracies are $q_1=0.5001$, $q_2=0.51$, and $q_3=0.61$.
A delegation from agent 1 to 2 is locally positive as $q_2\cdot p_{1,2}+(1-q_2)\cdot (1-p_{1,2}) = 0.501 > q_1$.
Furthermore, a delegation from 2 to 3 is locally positive as $q_3\cdot p_{2,3}+(1-q_3)\cdot (1-p_{2,3}) = 0.511 > q_2$.
However, the resulting delegation from $1$ to $3$ is not positive since $q_3\cdot p_{1,3}+(1-q_3)\cdot (1-p_{1,3}) = 0.412$.
\end{example}

%%%%%%%%%%%%%%%%%%%%%%%%%%%%%%%%%

\subsection{Quality of Equilibria in Delegation Games}

In delegation games players are motivated to maximize the tradeoff between the accuracy they acquire and the effort they spend for it. A natural measure for the quality of a delegation profile is, therefore, how accurate or informed a random voter becomes as a result of the delegations in the profile, that is, the average accuracy (i.e., $\aac({\bf d}) = \frac{1}{n}\sum_{i\in N} q^*_i({\bf d})$) players enjoy in that profile. One can also consider the utilitarian social welfare $\sw({\bf d}) = \sum_{i \in N} u_i({\bf d})$ of a delegation profile ${\bf d}$. This relates to average accuracy as follows: 
%Clearly average accuracy is directly related to utilitarian social welfare (i.e., %$\sw({\bf d}) = \sum_{i \in N} u_i({\bf d})$) in delegation games, as follows: 
\[
\aac({\bf d}) = \frac{\sw({\bf d}) + \sum_{i \mid d(i) = i} e_i}{n}.
\]
%That is, average accuracy is average %social welfare where effort of direct %voters is `reimbursed'. 

It can immediately be noticed that equilibria do not necessarily maximize 
%the quality of the group decision measured as 
average accuracy.
%, measured as average accuracy (i.e., $\frac{1}{n}\sum_{i\in N} q_i$).\footnote{This, however, should not come as a surprise %as in delegation games utility is defined in terms of accuracy, and it is well-known that in general NE do not maximize %social welfare (sum of individual utilities).} 
%that equilibria in delegation games do not necessarily achieve a good average accuracy. 
On the contrary, in the following example NE yields an average accuracy 
%(and utility) 
of close to $0.5$, whereas an average accuracy 
%(and utility) 
of almost $1$ is achievable.

\begin{example} \label{ex:quality}
Consider an $n$-player delegation game where all players have the same type and %everybody is connected except for player 1, which cannot delegate to anyone; this is specified by the network $R$. 
$(i,j) \in R$ for all $j$ and all $i>1$, i.e., player 1 is a sink in $R$ and cannot delegate to anyone, but all others can delegate to everyone. 
Further, we have $e_1=0$ and $e_i=0.5-\epsilon$ for $i\geq 2$. The respective accuracies are $q_1=0.5+2\epsilon$ and $q_i=1$.
If player $i\geq 2$ does not delegate, her utility is $0.5+\epsilon$. Hence, it is always more desirable to delegate to player $1$ (which yields a utility of $0.5+2\epsilon$ for $i$). Consider now the profiles in which all players delegate to player $1$ (either directly or transitively). Player $1$ can only vote directly (with utility $0.5+2\epsilon$). All such profiles are NE with average accuracy $0.5+2\epsilon$.
If, however, some player $j\geq 2$ chose to vote herself, all players (except $1$) would delegate to $j$ thereby obtaining an average accuracy of $1-\frac{0.5-2\epsilon}{n}$, which converges to $1$ for $n\to\infty$. This is not a NE, as $j$ could increase her utility by delegating to $1$.
\end{example}

The findings of the example can be made more explicit by considering a variant of the price of anarchy for delegation games, based on the above notion of average accuracy. That is, for a given delegation game $G$, the price of anarchy ($\PoA$) of $G$ is given by
\[
\PoA(G) = \frac{\max_{{\bf d} \in N^n} \aac({\bf d})}{\min_{{\bf d} \in \NE(G)} \aac({\bf d})},
\]
where $\NE(G)$ denotes the set of pure-strategy NE of $G$.

\begin{fact}
%In delegation games, the 
PoA is bounded below by $1$ and above by $2$ (see Apendix~\ref{appendix:proof}).
\end{fact}

An informative performance metrics for liquid democracy is the difference between the group accuracy after delegations versus the group accuracy achieved by direct voting. This measure, called {\em gain}, was introduced and studied by \citet{kahng18liquid}. Here we adapt it to our setting as follows:
%\footnote{Other variants are of course possible.}
$
    \g(G)  = \left(\min_{{\bf d} \in \NE(G)} \aac({\bf d})\right) - \bar{q}
$   
where $\bar{q} = \aac(\tuple{1, \ldots, n})$.
That is, the gain in the delegation game $G$ is the average accuracy of the worst NE minus the average accuracy of 
%the status quo, that is, 
the profile in which no voter delegates. It turns out that the full performance range is possible:
\begin{fact}
%In delegation games, gain 
$\g$ is bounded below by $-0.5$ and above by $0.5$ (see Apendix~\ref{appendix:proof}).
\end{fact}

The above bounds for PoA and gain provide only a very partial picture of the performance of liquid democracy when modeled as a delegation game. The next section provides a more fine-grained perspective on the effects of delegations. %in liquid democracy.
%, and rely on arguably artificial and %unrealistic examples. 
%Section \ref{sec:simulation} will %complete this picture by studying more %realistic instances of delegation games %through simulations.

%%%%%%%%%%%%%%%%%%%%%%%%%%%%%%%%%%%%%%%%%%%%%%%%%%%%%%%%%%%%%%%%%%%%%%%%%%%%%%%%%%%%%%%%%%%%%%%

\section{Simulations} %\label{sec:simulation}

We simulate the delegation game described above in a variety of settings. We restrict ourselves to homogeneous games. This allows us to relate our results to those of \citet{kahng18liquid}. Our experiments serve to both verify and extend the theoretical results of the previous section. In particular we simulate the best response dynamics employed in the proof of Theorem \ref{thm:delegation-games-NE} and show that these dynamics converge even in the setting with effort, which we could not establish analytically.
%\footnote{Cf. footnote \ref{footnote:homogeneous}.} 
In addition, we investigate the dynamics of a one-shot game scenario, in which all agents need to select their proxy simultaneously at once.
%with only local information and without %communication.
%
%$\tau$, to simplify the analysis and allow for %a direct comparison with their findings.

%\vspace{-0.3cm}

\paragraph{Setup}
We generate graphs of size $N=250$ of each of the four topologies \textit{random}, \textit{regular}, \textit{small world}, and \textit{scale free}, for different average degrees, while ensuring that the graph is connected.
%and thus all nodes are in $R^*$.
%\footnote{We assume undirected graphs as social ties are usually mutual.}
Agents' individual accuracy and effort are initialized randomly with $q_i \in [0.5,1]$ and $q_i - e_i \geq 0.5$. %\note[Martin]{Daan, please check - are $q_i \in [0.5,1]$ and $q_i - e_i \geq 0.5$ indeed our assumptions?}.
We average results over 2500 simulations for each setting (25 randomly generated graphs $\times$ 100 initializations). Agents correctly observe their own accuracy and effort, and the accuracy of their neighbors. The game is homogeneous, so proximities are $1$.

Each agent $i$ selects from her neighborhood $R(i)$ (which includes $i$ herself) the agent $j$ that maximizes her expected utility following \eqref{eq:utility}. We compare two scenarios. The {\em iterated best response} scenario follows the procedure of the proof of Theorem~\ref{thm:delegation-games-NE}, in which agents sequentially update their proxy to best-respond to the current profile ${\bf d}$ using knowledge of their neighbors' accuracy $q_i^*({\bf d})$. In the {\em one-shot game} scenario all agents choose their proxy only once, do so simultaneously, and based only on their neighbors' accuracy. The latter setup captures more closely the epistemic limitations that agents face in liquid democracy.

%the initial accuracies in their %neighborhood. 
%In both cases, the resulting profile ${\bf d}$ is applied recursively to find $d_i^*$ for all agents $i$.

\subsection{Iterated Best Response Dynamics}

These experiments complement our existence theorems. They offer insights into the effects of delegations on average voter's accuracy in equilibrium, and on the effects of different network structures on how such equilibria are achieved.

We initialize $q_i \sim \mathcal{N}(0.75,0.05)$ and first investigate the case in which $e_i = 0$ for all $i$ (effortless voting). Across all combinations of network types and average degrees ranging from 4 to 24, we find that the best response dynamics converges, as predicted by Theorem~\ref{thm:delegation-games-NE}, and does so optimally with $d_i^* = \argmax_{j} q_j$ for all $i$. We observe minimal differences between network types, but see a clear inverse relation between average degree and the %speed of convergence
number of iterations required to converge 
%(Figure~\ref{fig:iter_br_n250_q0.75_e0}). 
(Table~\ref{tab:convergence_rate}, top). Intuitively, more densely connected networks facilitate agents in identifying their optimal proxies (further details are provided in Appendix~\ref{appendix:simulation}).

% \begin{figure}[tb]
%     \centering
%     \includegraphics[width=\linewidth]{{iter_br_n250_q0.75_e0}.pdf}
%     \caption{Total number of best response updates of the individual agents (left) and corresponding number of full iterations over the network  (right) for different network types and degrees in the effortless setting.}
%     \label{fig:iter_br_n250_q0.75_e0}
%     \vspace{-0.2cm}
% \end{figure}

%The larger differences between the required number of best response updates for lower degree graphs of different types (e.g. degree 4) coincide with differences between the mean distance between nodes in those graphs: a shorter average distance yields a lower number of best response updates (Table~\ref{tab:dist_vs_br}). Intuitively this makes sense, as larger distances between nodes mean longer delegation chains. However, we have not yet conducted statistical tests to verify this hypothesis formally.

We accumulate results across all network types and compare the effortless setting to the case in which effort is taken into account. When we include effort $e_i \sim \mathcal{N}(0.025,0.01)$, we still observe convergence in all cases and, interestingly, the number of iterations required does not change significantly (Table~\ref{tab:convergence_rate}, bottom). Although the process no longer results in an optimal equilibrium, each case still yields a single guru $j$ with $q_j \approx \max_k q_k$ (less than $1\%$ error) for all $k \in N$. In this scenario, the inclusion of effort means that a best response update of agent $i$ no longer guarantees non-decreasing accuracy and utility for all other agents, which was a key property in the proof of Theorem \ref{thm:delegation-games-NE}. This effect becomes stronger as the average network degree increases, and as a result higher degree networks allow a greater discrepancy between the maximal average accuracy achievable and the average accuracy obtained at stabilization (Table~\ref{tab:br_accuracy}).%\footnote{The differences are statistically significant following a paired t-test with significance level 0.05.}

\begin{table}[tb]
    \caption{Total number of best response updates by individual agents and corresponding full passes over the network required for convergence. Reported are the mean (std.dev.) over all network types. \textit{Note}: not all agents update their delegation at each full pass, but any single update requires an additional pass to check whether the best response still holds.}
    \label{tab:convergence_rate}
    \begin{widetable}{\columnwidth}{l|cccccc}
    Degree & 4 & 8 & 12 & 16 & 20 & 24 \\
    \hline
    BR updates & 298.1 & 261.7 & 254.2 & 251.6 & 250.6 & 250.0 \\
    (effortless) & (18.2) & (11.1) & (6.9) & (4.5) & (3.3) & (2.6) \\[0.2em]
    Full passes & 3.6 & 3.0 & 2.9 & 2.8 & 2.7 & 2.5 \\
    (effortless) & (0.5) & (0.1) & (0.2) & (0.4) & (0.5) & (0.5) \\[0.2em]
    \hline
    BR updates & 294.7 & 259.4 & 252.9 & 250.9 & 250.2 & 249.9 \\
    (with effort) & (18.4) & (10.6) & (6.6) & (4.8) & (4.9) & (4.3) \\[0.2em]
    Full passes & 3.6 & 3.0 & 2.8 & 2.6 & 2.4 & 2.4 \\
    (with effort) & (0.5) & (0.3) & (0.6) & (0.8) & (0.9) & (1.0)
    \end{widetable}
\end{table}

\begin{table}[tb]
    \caption{Comparing the maximum accuracy across all agents and the mean accuracy under delegation ${\bf d}$ for different network degrees, averaged across network types. The differences are statistically significant (paired t-test, $p=0.05$).} 
    \label{tab:br_accuracy}
    \begin{widetable}{\columnwidth}{l|cccccc}
    Degree & 4 & 8 & 12 & 16 & 20 & 24 \\
    \hline
    $\max_j q_j$ & 0.8908 & 0.8908 & 0.8904 & 0.8909 & 0.8904 & 0.8910 \\
    $\bar{q}^*({\bf d})$ & 0.8906 & 0.8903 & 0.8897 & 0.8899 & 0.8890 & 0.8893
    \end{widetable}
    %\vspace{-0.4cm}
\end{table}

In lower degree graphs (e.g. degree 4) we further observe differences in convergence speed between the four different network types which coincide with differences between the average path lengths in those graphs: a shorter average distance between nodes yields a lower number of best response updates.
%\footnote{More detailed results supporting this finding are presented in the supplementary material, Appendix %B.} 
This is intuitive, as larger distances between nodes mean longer delegation chains, but we have not yet conducted statistical tests to verify this hypothesis.
%formally.

%$\max q_k$ and $\bar{q}^*({\bf d})$, the network %average accuracy under ${\bf d}$.

%It is important to note here that an %optimal best response equilibrium
%in terms of $\bar{q}^*({\bf d})$ 
%does not necessarily yield optimal %group's accuracy. In particular, the %probability of a correct vote now solely %depends on the single selected guru $j$ %and is thus $\sim q_j$.

%(cf. Condorcet's jury theorem %\citep{grofman83thirteen}). 
%In contrast, by Condorcet's jury theorem %\citep{grofman83thirteen}, the probability of a %correct majority vote under direct democracy %approaches 1 as $N \rightarrow \infty$ (since %we assume $q_i > 0.5$ for all $i$). 

%In the following we investigate whether %one-shot local dynamics can provide a middle %ground, since a larger number of gurus can be %expected.
%\note[Daan]{Not sure if this still makes %sense? I thought Condorcet's theorem is %about large numbers of voters, here we're %talking about just one, i.e. %dictatorship.}

\subsection{One-Shot Delegation Games}

Here we study one-shot interactions in a delegation game:
all agents select their proxy (possibly themselves) simultaneously among their neighbors; no further response is possible. This contrasts the previous scenario in which agents could iteratively improve their choice based on the choices of others.
While \citet{kahng18liquid} study a probabilistic model, we instead assume that agents deterministically select as proxy the agent $j \in R(i)$ that maximizes their utility, as above. We compare $\bar{q}$ and $\bar{q}^*$ (the average network accuracy without and with delegation, respectively), as well as the probability of a correct majority vote under both direct democracy $P_D$ and liquid democracy $P_L$ where gurus carry as weight the number of agents for whom they act as proxy. The difference $P_L - P_D$ is again similar to the notion of {\em gain} \citep{kahng18liquid}. In line with Condorcet's jury theorem (see for instance \citealt{grofman83thirteen}) $P_D \rightarrow 1$ as $N \rightarrow \infty$, and indeed for $N=250$ we obtain $P_D \approx 1$.

First we again look at the effortless setting. Figure~\ref{fig:acc_maj_n250} (top) shows both metrics for the four different network types and for different average degrees. 
\begin{figure}[tb]
	\centering
	\includegraphics[width=0.9\linewidth]{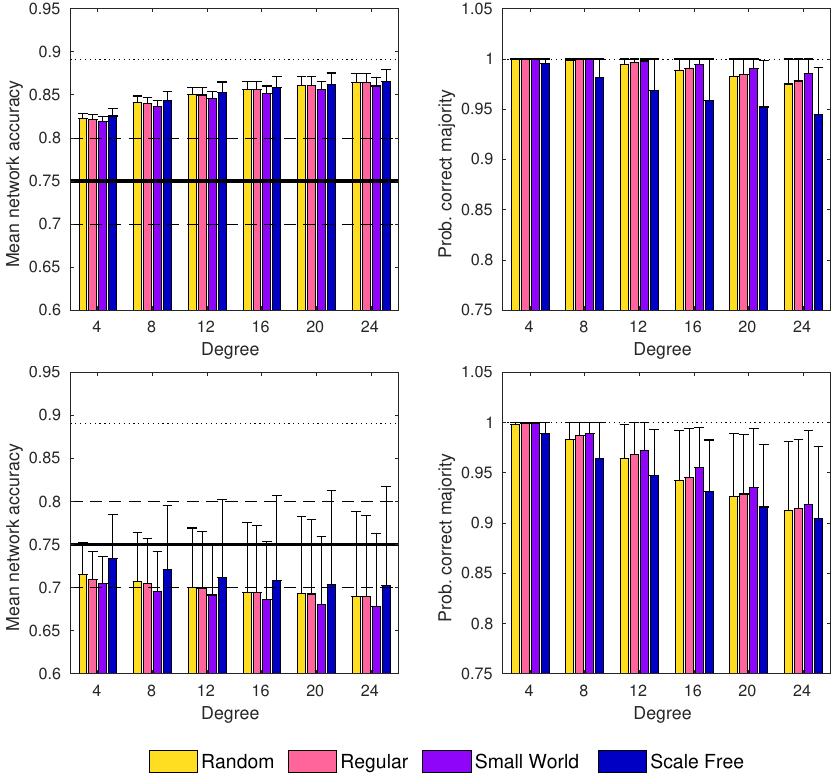}
	\caption{{\bf Top:} Without effort. {\bf Bottom:} With effort $e_i \sim \mathcal{N}(0.025, 0.01)$. {\bf Left:} mean accuracy under liquid democracy, $\bar{q}^*({\bf d})$. The solid (dashed) line shows the mean (std. dev.) of the initial accuracy $q$; the dotted line shows $\max_i q_i$. {\bf Right:} probability of a correct majority vote under ${\bf d}$.}
	\label{fig:acc_maj_n250}
	%\vspace{-0.2cm}
\end{figure}
We observe that while $\bar{q}^*({\bf d})$ increases as the network degree increases (and in fact is always higher than $\bar{q}$ without delegation), the probability of a correct majority outcome, $P_L$, simultaneously decreases. This confirms the analysis of \citet{kahng18liquid}. We also observe that the number of gurus decreases exponentially as the degree increases (Figure~\ref{fig:perc_dist_n250_q0.75_e0}, left). Simply put, giving all voting weight to a small group of gurus increases the chance of an incorrect majority vote, assuming that gurus have a less than perfect accuracy.

\begin{figure}[tb]
    \centering
    \includegraphics[width=0.9\linewidth]{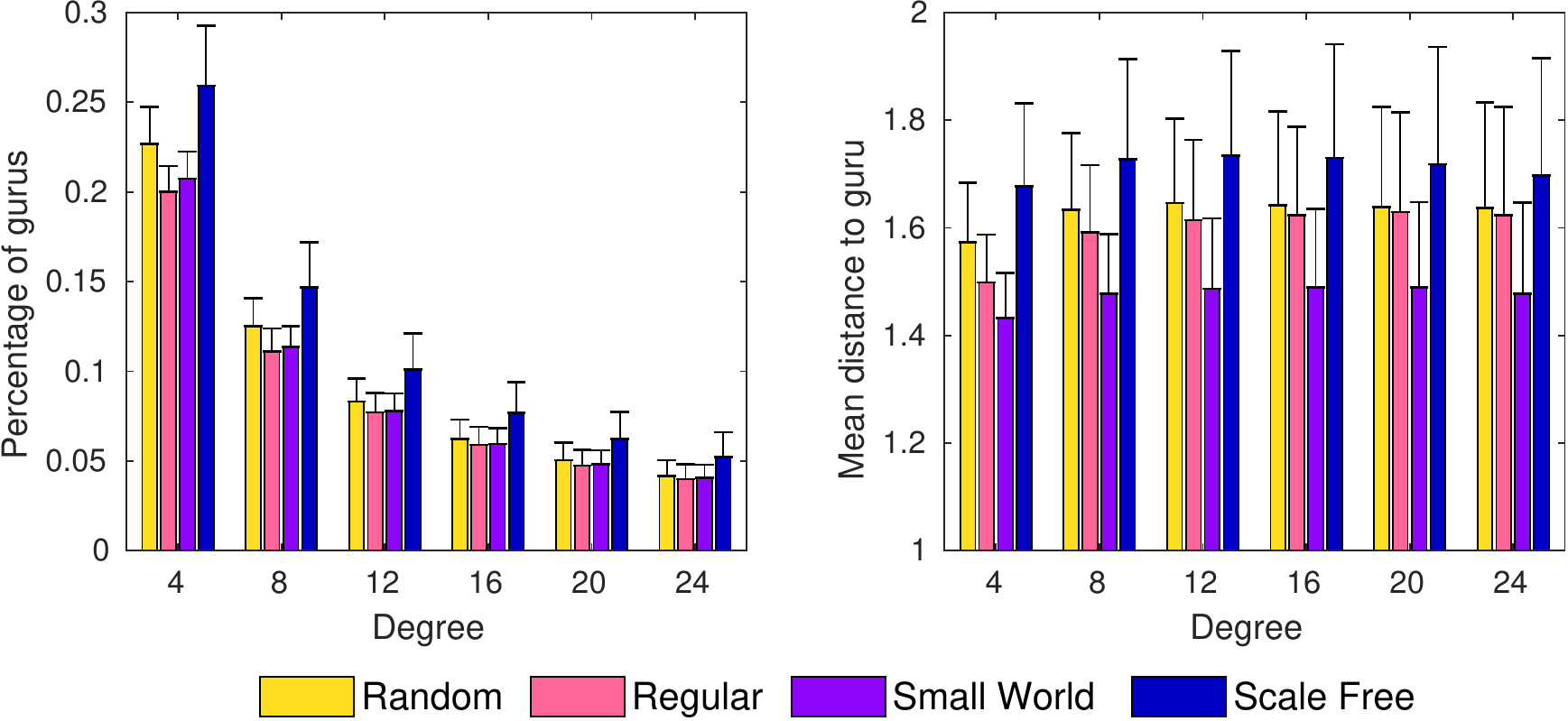}
    \caption{Percentage of guru nodes under ${\bf d}$ (left) and mean distance between (non-guru) nodes and their gurus (right).}
    \label{fig:perc_dist_n250_q0.75_e0}
    %\vspace{-0.4cm}
\end{figure}

When we include effort (Figure~\ref{fig:acc_maj_n250}, bottom), thereby moving away from the model of \citet{kahng18liquid}, we observe a drastic decrease in average network accuracy combined with a lower probability of a correct majority outcome under liquid democracy, with both decreasing as network degree increases. The main reason is the existence of delegation cycles in this case. This contrasts the best response setting above where agents could iteratively reconsider their choice of proxy and thus avoid cycles. Now, even with relatively low effort (mean 0.025), up to {\em half} of all agents are stuck in a cycle (and thereby fail to pass information about their type) when the degree increases. %\footnote{When we omit agents in cycles altogether the mean accuracy of the remaining agents increases significantly, but the probability of a correct majority increases only slightly and does not reach the level of effortless voting. The reason is that a smaller number of gurus remains, with consequently a higher probability of an incorrect majority.} 
This confirms results on the probability of delegation cycles from \citet{christoff17binary} and stresses the importance of cycle resolution in concrete implementations of liquid democracy. 
% such as Liquid Feedback.

%As a result, those agents cast a random vote by %our model (Eq.~\ref{eq:accuracy_deterministic}). 

%Thus, an important consideration is how to %resolve cycles; this fact has so far been %underrepresented in literature.

Figure~\ref{fig:acc_maj_n250} further highlights differences between the four network types. Scale free networks yield a lower probability of a correct majority outcome across all degrees, as well as a larger number of gurus with a lower average accuracy and longer delegation chains (Figure~\ref{fig:perc_dist_n250_q0.75_e0}, right). Intuitively, this indicates that one-shot interactions in scale free networks are more likely to end up in a local optimum. In contrast, small world networks have short average distances 
%between nodes, 
and thus agents are more likely to be close to their optimal guru. 
Finally, our experiments highlight a key feature of liquid democracy: the trade-off between a reduction in (total) effort against a loss in average voting accuracy.

%A more detailed study of these hypotheses is %left for future work.

%%%%%%%%%%%%%%%%%%%%%%%%%%%%%%%%%%%%%%%%%%%%%%%%%%%%%%%%%%%%%%%%%%%%%%%%%%%%%%%%%%%%

\section{Conclusions and Future Work} %\label{sec:conclusions}

The paper introduced delegation games as a first game-theoretic model of liquid democracy. Both our theoretical and experimental results showed that voting effort is a key ingredient for understanding how delegations form and what their effects are.
%the behavior and effects of delegations.
Our empirical findings provided further insights 
%that are hard to capture by a purely %theoretical analysis, especially %concerning 
into the influence of interaction networks on the quality of collective decisions in liquid democracy. 

The paper opens up several directions of research.
%Many directions for future work present %themselves. 
%We mention two of many natural directions for future research.
A general NE existence theorem is the main open question. Our model can then be generalized in many directions, e.g.: by making agents' utility dependent on voting outcomes; by dropping the independence assumption on agents' types; or by assuming the voting mechanism has better than $0.5$ accuracy in identifying the types of agents involved in cycles.

\bibliography{liquid_bib.bib}

\begin{thebibliography}{}

\bibitem[\protect\citeauthoryear{Alger}{Alger}{2006}]{Alger_2006}
Alger, D. (2006).
\newblock {Voting by proxy}.
\newblock {\em Public Choice}, {\em 126\/}(1-2), 1--26.

\bibitem[\protect\citeauthoryear{Barab{\'a}si \& Albert}{Barab{\'a}si \&
  Albert}{1999}]{barabasi1999emergence}
Barab{\'a}si, A.-L. \& Albert, R. (1999).
\newblock Emergence of scaling in random networks.
\newblock {\em science}, {\em 286\/}(5439), 509--512.

\bibitem[\protect\citeauthoryear{Blum \& Zuber}{Blum \&
  Zuber}{2016}]{blum2016liquid}
Blum, C. \& Zuber, C.~I. (2016).
\newblock Liquid democracy: Potentials, problems, and perspectives.
\newblock {\em Journal of Political Philosophy}, {\em 24\/}(2), 162--182.

\bibitem[\protect\citeauthoryear{Boella, Francis, Grassi, Kistner, Nitsche,
  Noskov, Sanasi, Savoca, Schifanella \& Tsampoulatidis}{Boella
  et~al.}{2018}]{boella2018wegovnow}
Boella, G., Francis, L., Grassi, E., Kistner, A., Nitsche, A., Noskov, A.,
  Sanasi, L., Savoca, A., Schifanella, C., \& Tsampoulatidis, I. (2018).
\newblock {WeGovNow}: A map based platform to engage the local civic society.
\newblock In {\em {WWW} '18 Companion: The 2018 Web Conference Companion},
  (pp.\ 1215--1219)., Lyon, France. International World Wide Web Conferences
  Steering Committee.

\bibitem[\protect\citeauthoryear{Boldi, Bonchi, Castillo \& Vigna}{Boldi
  et~al.}{2011}]{boldi2011viscous}
Boldi, P., Bonchi, F., Castillo, C., \& Vigna, S. (2011).
\newblock Viscous democracy for social networks.
\newblock {\em Communications of the ACM}, {\em 54\/}(6), 129--137.

\bibitem[\protect\citeauthoryear{Brandt, Conitzer, Endriss, Lang \&
  Procaccia}{Brandt et~al.}{2016}]{brandt16handbook}
Brandt, F., Conitzer, V., Endriss, U., Lang, J., \& Procaccia, A.~D. (Eds.).
  (2016).
\newblock {\em Handbook of Computational Social Choice}.
\newblock Cambridge University Press.

\bibitem[\protect\citeauthoryear{Brill}{Brill}{2018}]{brill2018interactive}
Brill, M. (2018).
\newblock Interactive democracy.
\newblock In {\em Proceedings of the 17th International Conference on
  Autonomous Agents and MultiAgent Systems}, (pp.\ 1183--1187). International
  Foundation for Autonomous Agents and Multiagent Systems.

\bibitem[\protect\citeauthoryear{Christoff \& Grossi}{Christoff \&
  Grossi}{2017}]{christoff17binary}
Christoff, Z. \& Grossi, D. (2017).
\newblock Binary voting with delegable proxy: An analysis of liquid democracy.
\newblock In {\em Proceedings of {TARK'17}}, volume 251, (pp.\ 134--150).
  EPTCS.

\bibitem[\protect\citeauthoryear{Cohensius, Mannor, Meir, Meirom \&
  Orda}{Cohensius et~al.}{2017}]{cohensius2017proxy}
Cohensius, G., Mannor, S., Meir, R., Meirom, E., \& Orda, A. (2017).
\newblock Proxy voting for better outcomes.
\newblock In {\em Proceedings of the 16th Conference on Autonomous Agents and
  MultiAgent Systems}, (pp.\ 858--866). International Foundation for Autonomous
  Agents and Multiagent Systems.

\bibitem[\protect\citeauthoryear{Dodgson}{Dodgson}{1884}]{dodgson84principles}
Dodgson, C.~L. (1884).
\newblock {\em The Principles of Parliamentary Representation}.
\newblock Harrison and Sons.

\bibitem[\protect\citeauthoryear{Endriss}{Endriss}{2016}]{endriss16judgment}
Endriss, U. (2016).
\newblock Judgment aggregation.
\newblock In F.~Brandt, V.~Conitzer, U.~Endriss, J.~Lang, \& A.~D. Procaccia
  (Eds.), {\em Handbook of Computational Social Choice}. Cambridge University
  Press.

\bibitem[\protect\citeauthoryear{Erd{\"o}s \& R{\'e}nyi}{Erd{\"o}s \&
  R{\'e}nyi}{1959}]{erdos1959random}
Erd{\"o}s, P. \& R{\'e}nyi, A. (1959).
\newblock On random graphs, {I}.
\newblock {\em Publicationes Mathematicae (Debrecen)}, {\em 6}, 290--297.

\bibitem[\protect\citeauthoryear{G{\"o}lz, Kahng, Mackenzie \&
  Procaccia}{G{\"o}lz et~al.}{2018}]{golz2018fluid}
G{\"o}lz, P., Kahng, A., Mackenzie, S., \& Procaccia, A.~D. (2018).
\newblock The fluid mechanics of liquid democracy.
\newblock In {\em WADE'18, arXiv:1808.01906}.

\bibitem[\protect\citeauthoryear{Green-Armytage}{Green-Armytage}{2015}]{green2015direct}
Green-Armytage, J. (2015).
\newblock Direct voting and proxy voting.
\newblock {\em Constitutional Political Economy}, {\em 26\/}(2), 190--220.

\bibitem[\protect\citeauthoryear{Grofman, Owen \& Feld}{Grofman
  et~al.}{1983}]{grofman83thirteen}
Grofman, B., Owen, G., \& Feld, S. (1983).
\newblock Thirteen theorems in search of truth.
\newblock {\em Theory and Decision}, {\em 15}, 261--278.

\bibitem[\protect\citeauthoryear{Grossi \& Pigozzi}{Grossi \&
  Pigozzi}{2014}]{Grossi_2014}
Grossi, D. \& Pigozzi, G. (2014).
\newblock {\em {Judgment Aggregation: A Primer}}.
\newblock Synthesis Lectures on Artificial Intelligence and Machine Learning.
  Morgan {\&} Claypool Publishers.

\bibitem[\protect\citeauthoryear{Jackson}{Jackson}{2008}]{jackson08social}
Jackson, M.~O. (2008).
\newblock {\em Social and Economic Networks.}
\newblock Princeton University Press.

\bibitem[\protect\citeauthoryear{Kahng, Mackenzie \& Procaccia}{Kahng
  et~al.}{2018}]{kahng18liquid}
Kahng, A., Mackenzie, S., \& Procaccia, A. (2018).
\newblock Liquid democracy: An algorithmic perspective.
\newblock In {\em Proc. 32nd AAAI Conference on Artificial Intelligence
  (AAAI'18)}.

\bibitem[\protect\citeauthoryear{Kling, Kunegis, Hartmann, Strohmaier \&
  Staab}{Kling et~al.}{2015}]{kling2015voting}
Kling, C.~C., Kunegis, J., Hartmann, H., Strohmaier, M., \& Staab, S. (2015).
\newblock Voting behaviour and power in online democracy: A study of
  liquidfeedback in germany's pirate party.
\newblock In {\em Ninth International AAAI Conference on Web and Social Media}.

\bibitem[\protect\citeauthoryear{Miller}{Miller}{1969}]{miller1969program}
Miller, J.~C. (1969).
\newblock A program for direct and proxy voting in the legislative process.
\newblock {\em Public choice}, {\em 7\/}(1), 107--113.

\bibitem[\protect\citeauthoryear{Skowron, Lackner, Brill, Peters \&
  Elkind}{Skowron et~al.}{2017}]{skowron17proportional}
Skowron, P., Lackner, M., Brill, M., Peters, D., \& Elkind, E. (2017).
\newblock Proportional rankings.
\newblock In {\em Proceedings of the Twenty-Sixth International Joint
  Conference on Artificial Intelligence, {IJCAI} 2017, Melbourne, Australia,
  August 19-25, 2017}, (pp.\ 409--415).

\bibitem[\protect\citeauthoryear{Tullock}{Tullock}{1992}]{Tullock_1992}
Tullock, G. (1992).
\newblock {Computerizing politics}.
\newblock {\em Mathematical and Computer Modelling}, {\em 16\/}(8-9), 59--65.

\bibitem[\protect\citeauthoryear{Watts \& Strogatz}{Watts \&
  Strogatz}{1998}]{watts1998collective}
Watts, D.~J. \& Strogatz, S.~H. (1998).
\newblock Collective dynamics of ‘small-world’ networks.
\newblock {\em Nature}, {\em 393\/}(6684), 440.

\end{thebibliography}
\bibliographystyle{newapa}

%%%%%%%%%%%%%%%%%%%%%%%%%%%%%%%%%%%%%%%%%%%%%%%%

\newpage

\appendix

\section{Proofs of Proposition 1, Fact 1, and Fact 2}
\label{appendix:proof}

Here we provide the proofs op Proposition~1, Fact~1, and Fact~2.

\begin{proof}[Proof of Proposition 1]
%Let ${\bf d}=\tuple{d_1,\dots,d_n}$.
%If $d(t)=t$, then $q^*_s({\bf d}') = q_t\cdot p_{s,t} + (1-q_t)\cdot (1-p_{s,t}) > q_s$, since the delegation from $s$ to $t$ is locally rational. This proves the statement for delegation chains of length 1.
First we show that $q^*_s({\bf d}') > q^*_s({\bf d}) = q_s$.
If $d_t=t$, the statement holds since the delegation from $s$ to $t$ is locally positive.
Otherwise, let $d_t=r$.
%Assume that we have shown the statement for delegation chains of length $k$.
%Let $s,t,x_1,\dots,x_{k-2},r$ be a delegation chain of length $k+1$ with $x_3=d_t$, $x_4=d_{x_3}$, \dots, $r=d_{x_{k-2}}$.
Since ${\bf d}$ is positive, we know that 
$$
q^*_t({\bf d}) = q_r\cdot p_{t,r} + (1-q_r)\cdot (1-p_{t,r}) >q_t.
$$
Now let $x_i = {\mathbb P}(\tau(i) = 1)$ for $i \in \set{s,t,r}$.
Since types are independent random variables, it holds that $p_{i,j} = x_ix_j+(1-x_i)(1-x_j)$.
Using this fact, we can verify that
\begin{align}
\begin{split}
p_{s,r} = {} &
%\P(\tau(s)=\tau(r))
%&=\P(\tau(s)=\tau(t))\P(\tau(t)=\tau(r))+\\
%&\quad+\P(\tau(s)\neq\tau(t))\P(\tau(t)\neq\tau(r))\\
p_{s,t}p_{t,r}+(1-p_{s,t})(1-p_{t,r}) \\ & \quad + (- 2 (2x_s-1) \cdot (2 x_r - 1) \cdot (x_t - 1) x_t).
\end{split}
%p_{s,r} = p_{s,t}p_{t,r}+(1-p_{s,t})(1-p_{t,r}) + (- 2 (2x_s-1) \cdot (2 x_r - 1) \cdot (x_t - 1) x_t).
\label{eq:fact}
\end{align}
%where $x_i = \P(\tau(i) = 1)$ for $i \in \set{s,t,r}$. 
We now want to prove that
\begin{align}- 2 (2x_s-1) \cdot (2 x_r - 1) \cdot (x_t - 1) x_t \geq 0.
\label{eq:intermediateinequ}
\end{align}
%To prove \eqref{eq:intermediateinequ},
First note that a locally positive delegation from $i$ to $j$ implies that $p_{i,j}\geq 0.5$ (since $q_i,q_j\geq 0.5$); hence $p_{s,t}\geq 0.5$ and $p_{t,r}\geq 0.5$. 
Second, observe that for any $a,b\in[0,1]$, if $ab+(1-a)(1-b)\geq 0.5$, then either $a,b\in[0,0.5]$ 
or $a,b\in[0.5,1]$.
Now, since $p_{s,t} = x_sx_t+(1-x_s)(1-x_t)\geq 0.5$ and $p_{t,r} = x_tx_r+(1-x_t)(1-x_r)\geq 0.5$, we conclude that either $x_s,x_t,x_r \in [0,0.5]$ or $x_s,x_t,x_r \in [0.5,1]$. 
In both cases \eqref{eq:intermediateinequ} holds.

From \eqref{eq:intermediateinequ} and \eqref{eq:fattolino} it follows that:%
\begin{equation}
\begin{split}
p_{s,r} & \geq p_{s,t}p_{t,r}+(1-p_{s,t})(1-p_{t,r}) \\ & = 2p_{s,t}p_{t,r}-p_{s,t}-p_{t,r}+1.
\end{split}
\label{eq:fattolino}
\end{equation}
%Consequently, since $q_t\cdot p_{s,t} + (1-q_t)\cdot (1-p_{s,t})>q_s\geq 0.5$, $q_t$
%, by \eqref{eq:accuracy}, the assumptions made on ${\bf d}$ and ${\bf d}'$, and the fact that 
%, it follows that  $x_i \geq 0.5$ or $i \in \set{s,t,r}$ and therefore

Thus we obtain:
%\begin{align*}
%q^*_s({\bf d}') &= q_r p_{s,r} + (1- q_r) (1-p_{s,r}) \\ 
%&= p_{s,r} (2q_r  -1 ) - q_r + 1 \\
%& \geq (p_{s,t}+p_{t,r}-1)(2q_r  -1 ) - q_r + 1\\
%&= \underbrace{(2q_r p_{s,t} -  p_{s,t} - q_r + 1 )}_{q_r \geq q_t} + (2q_r p_{t,r} -  p_{t,r} - q_r + 1) - q_r\\
%& > q_s + q_t - q_r\\
%& \geq \underbrace{2q_r}_{> 1}(\underbrace{2p_{s,t}p_{t,r} - p_{t,r} - p_{s,t} +1}_\text{$\geq p_{s,r}$ by \eqref{eq:fattolino}}) -q_r - (2p_{s,t}p_{t,r} - p_{t,r} - p_{s,t} +1) + 1 & \\
%& = 4q_r p_{s,t} p_{t,r} - 2 q_r p_{t,r} - 2 q_r p_{s,t} + 2 q_r - q_r - 2 p_{s,t}p_{t,r} + p_{t,r} + p_{s,t} -1 + 1 \\
%& =  4q_r p_{s,t} p_{t,r} - 2 p_{s,t}p_{t,r}  - 2 q_r p_{s,t} + 2 p_{s,t} - 2 q_r p_{t,r} + q_r + p_{t,r} - 1 + 1 - p_{s,t} \\
%& = (\underbrace{2p_{s,t}-1}_{> 0})(\underbrace{2 q_r p_{t,r} - q_r - p_{t,r} + 1}_\text{$>q_t$ since ${\bf d}$ is positive}) + 1 -  p_{s,t} \\
%& > (2 p_{s,t} - 1) q_t + 1 -  p_{s,t} \\
%& = \underbrace{2 q_t p_{s,t}  - q_t - p_{s,t} + 1}_\text{$>q_s$ as ${\bf d}'_s = t$ is locally positive} \\ %&  \text{by}~q_r \in [0.5,1] \\
%& > q_s.
%\end{align*}
%OLD:
\begin{align*}
q^*_s({\bf d}') &= q_r p_{s,r} + (1- q_r) (1-p_{s,r}) \\ 
&= 2 q_r p_{s,r} - q_r - p_{s,r} + 1 \\
& \geq \underbrace{2q_r}_{\geq 1}(\underbrace{2p_{s,t}p_{t,r} - p_{t,r} - p_{s,t} +1}_\text{$\geq p_{s,r}$ by \eqref{eq:fattolino}}) - q_r - (2p_{s,t}p_{t,r} - p_{t,r} - p_{s,t} +1) + 1 \\
&= 4q_r p_{s,t} p_{t,r} - 2 q_r p_{t,r} - 2 q_r p_{s,t} + 2 q_r - q_r - 2 p_{s,t}p_{t,r} + p_{t,r} + p_{s,t} -1 + 1 \\
&=  4q_r p_{s,t} p_{t,r} - 2 p_{s,t}p_{t,r}  - 2 q_r p_{s,t} + 2 p_{s,t} - 2 q_r p_{t,r} + q_r + p_{t,r} - 1 + 1 - p_{s,t} \\
&= (\underbrace{2p_{s,t}-1}_{> 0})(\underbrace{2 q_r p_{t,r} - q_r - p_{t,r} + 1}_\text{$>q_t$ since ${\bf d}$ is positive}) + 1 -  p_{s,t} \\
&> (2 p_{s,t} - 1) q_t + 1 -  p_{s,t} \\
&= \underbrace{2 q_t p_{s,t}  - q_t - p_{s,t} + 1}_\text{$>q_s$ as ${\bf d}'_s = t$ is locally positive} \\ %&  \text{by}~q_r \in [0.5,1] \\
&> q_s.
\end{align*}
%\begin{align*}
%q^*_s({\bf d}') &= 2 q_r p_{s,r} - q_r - p_{s,r} + 1 \\
%&= 4q_r p_{s,t} p_{t,r} - 2 p_{s,t} p_{t,r}  - p_{s,t} q_r + 2p_{s,t} - 2 p_{t,r}q_r\\
%&\quad + q_r - p_{s,t} + p_{t,r}\\
%&= (2p_{s,t}-1)(\underbrace{2 q_r p_{t,r} - q_r - p_{t,r} + 1}_\text{$>q_t$ since ${\bf d}$ is positive}) + 1 -  p_{s,t} \\
%&> 2p_{s,t}- q_t - p_{s,t} + 1 > q_s.
%\end{align*}
We conclude that $q^*_s({\bf d}') > q_s$.
It remains to show that for an agent $z$ with $d_z=s$, it still holds that $q^*_z({\bf d}')> q_z$.
This can be shown by the same argument as above, now using the fact that the delegations from $z$ to $s$ and from $s$ to $r$ are positive.
Hence ${\bf d}'$ is positive.
\end{proof}

\begin{proof}[Proof of Fact 1]
For the upper bound of $2$, Example \ref{ex:quality} shows that maximal average accuracy can be made arbitrarily close to $1$, while the average accuracy of a worse NE can be made arbitrarily close to $0.5$. For the lower bound of $1$ it suffices to consider any delegation game were $R = \emptyset$ (that is, no delegation is possible). 
\end{proof}

\begin{proof}[Proof of Fact 2]
For the lower bound we can use again Example \ref{ex:quality}, where the average accuracy of the worse NE can be made arbitrarily close to $0.5$, and the average accuracy of direct voting (which given the structure of the example is maximized by direct voting) can be made arbitrarily close to $1$. 
For the upper bound, consider an effortless homogeneous game where $q_1 = 1$, and $q_i = 0.5 + \epsilon$ for any $i \neq 1$ in $N$. Assume furthermore that $R$ is such that any agent can delegate only to $1$. In such a game, the average accuracy of the unique NE is $1$, while the average accuracy of direct voting can be made arbitrarily close to $0.5$. We thereby obtain the desired bound of $0.5$.
\end{proof}

\section{Additional Simulation Results}
\label{appendix:simulation}

Detailed results for the ``Simulations'' subsection ``Iterated Best Response Dynamics''.  Figure~\ref{fig:iter_br_n250_q0.75_e0} shows convergence speed results in the effortless setting for different network types. The larger differences between the required number of best response updates for lower degree graphs of different types (e.g. degree 4) coincide with differences between the mean distance between nodes in those graphs: a shorter average distance yields a lower number of best response updates. The latter is shown in Table~\ref{tab:dist_vs_br}.

\begin{figure}[h!]
    \centering
    \includegraphics[width=\linewidth]{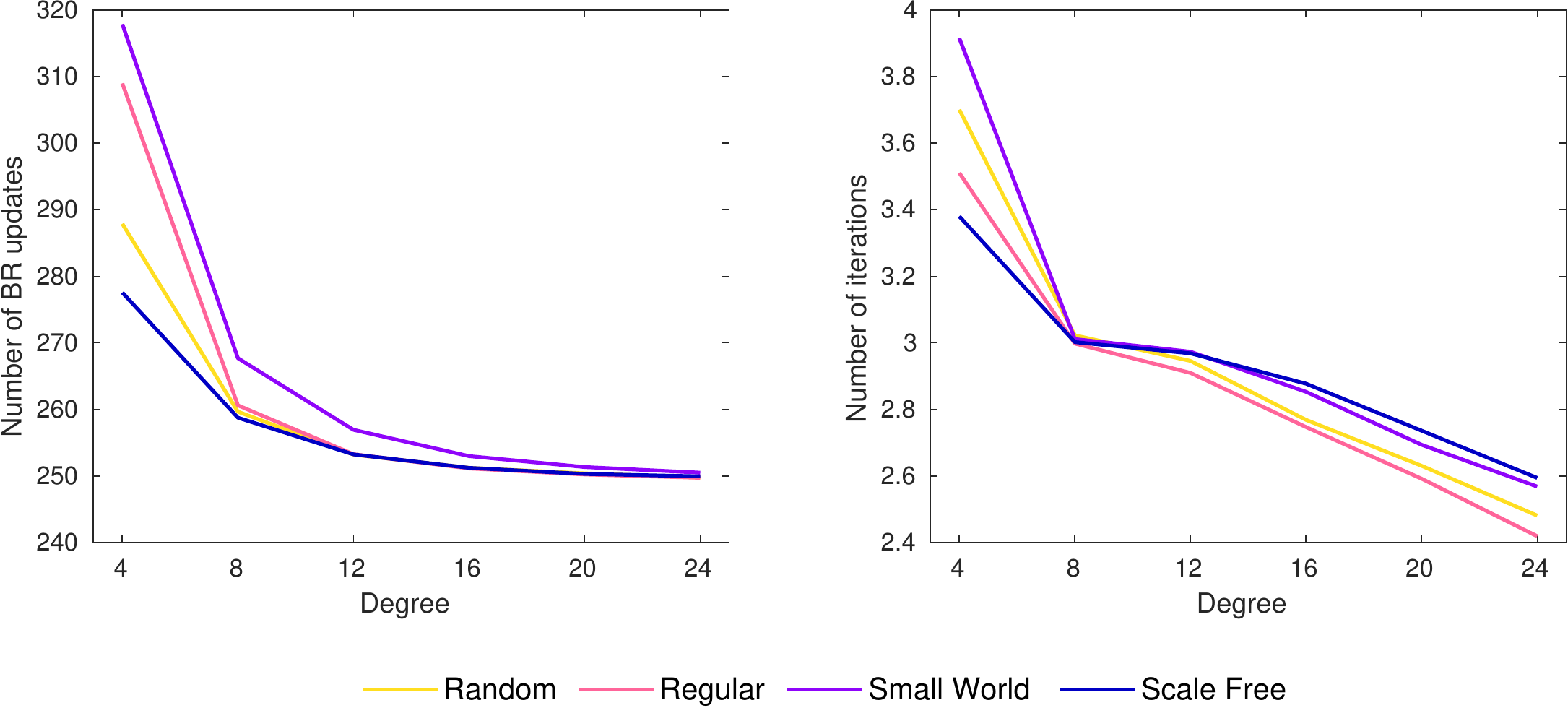}
    \caption{Total number of best response updates of the individual agents (left) and corresponding number of full iterations over the network  (right) for different network types and degrees in the effortless setting.}
    \label{fig:iter_br_n250_q0.75_e0}
\end{figure}

\begin{table}[h!]
    \centering
    \caption{Mean distance between nodes and required number of best response updates for degree-4 networks of different types, sorted from high to low.}
    \label{tab:dist_vs_br}
    \begin{widetable}{\columnwidth}{l|cccc}
    Type & Small World & Regular & Random & Scale Free \\
    \hline
    Mean distance & 4.97 & 4.39 & 4.03 & 3.41 \\
    BR updates & 317.88 & 308.98 & 287.88 & 277.55 \\
    \end{widetable}
\end{table}

\end{document}